\documentclass[review]{elsarticle}
\usepackage{hyperref}
\usepackage{amsmath, amsthm, amsfonts, amssymb, mathtools, graphicx, complexity}
\usepackage{subfig}
\usepackage{caption}               
\usepackage{float}
\usepackage{algcompatible}
\usepackage{algorithm, algpseudocode, float}
\usepackage{tikz}
\usetikzlibrary{positioning}
\usetikzlibrary{decorations.markings}
\usepackage{chngcntr}
\usepackage{setspace}

\newclass{\RBP}{RedBluePath}
\newclass{\EP}{EvenPath}
\newclass{\CBPL}{CBPL}
\newclass{\CSC}{CSC_1}
\newclass{\CSPACE}{CSPACE}
\newclass{\CSSPACE}{CSSPACE}
\newclass{\CTISP}{CTISP}
\newclass{\CBPSPACE}{CBPSPACE}
\newclass{\CNL}{CNL}
\newclass{\CL}{CL}
\newclass{\LCL}{LCL}
\newclass{\hamdist}{hamdist}
\newclass{\coCNL}{coCNL}
\newclass{\CUL}{CUL}
\newclass{\coCUL}{coCUL}
\newclass{\CNSPACE}{CNSPACE}
\newclass{\CUSPACE}{CUSPACE}
\newclass{\ZTIME}{ZTIME}
\newclass{\LCSPACE}{LCSPACE}
\newclass{\uniform}{uniform}

\theoremstyle{plain}
\newtheorem{thm}{Theorem}
\newtheorem{nlemma}[thm]{Lemma}
\theoremstyle{definition}
\newtheorem{defn}[thm]{Definition}

\date{ }

\journal{arXiv}

\begin{document}

\begin{frontmatter}

\title{Lossy Catalytic Computation}

\author[address1]{Chetan Gupta}
\ead{chetan.gupta@cs.iitr.ac.in}

\author[address2]{Rahul Jain\corref{mycorrespondingauthor}}
\cortext[mycorrespondingauthor]{Corresponding author}
\ead{rahul.jain@uni-siegen.de}

\author[address3]{Vimal Raj Sharma}
\ead{vimalraj@iitj.ac.in}

\author[address4]{Raghunath Tewari}
\ead{rtewari@cse.iitk.ac.in}

\address[address1]{Indian Institute of Technology Roorkee - Haridwar Highway, Roorkee, Uttarakhand, 247667, India}
\address[address2]{Universität Siegen, Germany}
\address[address3]{Indian Institute of Technology Jodhpur, N.H.- 62, Nagaur Road, Karwar, Jodhpur, 342030, India}
\address[address4]{Indian Institute of Technology Kanpur, Kalyanpur, Kanpur, 208016, India}

\begin{abstract}
A catalytic Turing machine is a variant of a Turing machine in which there exists an auxiliary tape in addition to the input tape and the work tape. This auxiliary tape is initially filled with arbitrary content. The machine can read and write on the auxiliary tape, but it is constrained to restore its initial content when it halts. Studying such a model and finding its powers and limitations has practical applications.

In this paper, we study catalytic Turing machines with $O(\log n)$-sized work tape and polynomial-sized auxiliary tape that are allowed to lose at most constant many bits of the auxiliary tape when they halt. We show that such catalytic Turing machines can only decide the same set of languages as standard catalytic Turing machines with the same size work and auxiliary tape.
\end{abstract}


\end{frontmatter}

\section{Introduction}
A $\emph{catalytic Turing machine}$ is a recently introduced model of computation by Buhrman et al. \cite{Buhrman2014} that has an auxiliary tape filled with arbitrary content in addition to the work tape of a standard Turing machine. The machine, during the computation, can read and write to the auxiliary tape, but when it halts, it is constrained to have the same content in the auxiliary tape as it had initially. Whether the catalytic Turing machine model is more powerful than the traditional Turing machine model is the holy grail of catalytic computation. Intuitively, the extra space might look useless; however, Buhrman et al. \cite{Buhrman2014} showed that catalytic Turing machines with $O(\log n)$ work space and $n^{O(1)}$ auxiliary space (Catalytic logspace, $\CL$) could solve problems which are not known to be solvable by standard Turing machines using $O(\log n)$ work space (Logspace, $\textsf{L}$). Specifically, they showed that $\CL$ contains the circuit class $\textsf{uniform}$-$\TC_1$, which is suspected to be different from $\textsf{L}$. They also proved that $\CL$ is contained in $\ZPP$, the set of languages decidable in expected polynomial time. Later, Buhrman et al. \cite{Buhrman2018} also defined the nondeterministic catalytic computational model. The complexity class $\CNL$ is the nondeterministic variant of $\CL$. They showed that, under a standard derandomization assumption, $\CNL$ is closed under complement. Many other researchers have studied the catalytic model as well \cite{Girard, PotechinP17, CM20, GJST19, DGJST20, BDS20, CookM21, CookM22, CookM24, Pyne24}. For an overview of the catalytic model, one can refer to the surveys by Kouck{\'{y}} \cite{Koucky16} and Mertz \cite{Mertz23}.

In this paper, we study whether allowing catalytic Turing machines to lose some auxiliary bits increases their power. A \emph{$k$-lossy catalytic Turing machine} is a catalytic Turing machine that is allowed to lose at most $k$ many arbitrary bits when it halts, where $k$ is a constant.  We prove that logspace $k$-lossy catalytic Turing machines are not more powerful than the standard logspace catalytic Turing machines. More specifically, in the logspace setting, we show that for every $k$-lossy catalytic Turing machine $\mathcal{M}$ there exists a normal catalytic Turing machine $\mathcal{M'}$ which can decide the language $L(\mathcal{M})$. We crucially use \emph{FKS hashing scheme} to achieve this result.

\section{Preliminaries}
\label{lossy:pre}
The deterministic catalytic Turing machine was formally defined by Buhrman et al. \cite{Buhrman2018} in the following way.

\begin{defn}
Let $\mathcal{M}$ be a deterministic Turing machine with three tapes: one input tape, one work tape, and one \textit{auxiliary tape}. $\mathcal{M}$ is said to be a \textit{deterministic catalytic Turing machine} using work space $s(n)$ and auxiliary space
$s_a(n)$ if for all inputs $x \in \{0, 1\}^n$ and auxiliary tape contents $w \in \{0, 1\}^{s_a(n)}$, the following three properties hold.
\begin{enumerate}
\item \textbf{Space bound.} The machine $\mathcal{M}$ uses space $s(n)$ on its work tape and space $s_a(n)$ on its auxiliary tape.
\item \textbf{Catalytic condition.} $\mathcal{M}$ halts with $w$ on its auxiliary tape.
\item \textbf{Consistency.} $\mathcal{M}$ either accepts $x$ for all choices of $w$ or it rejects for all choices of $w$.
\end{enumerate}
\end{defn}

For any catalytic Turing machine $\mathcal{M}$, input $x$, and auxiliary content $w$, $\mathcal{M}(x,w)$ denotes the computation of $\mathcal{M}$ on $x$ and $w$.

\begin{defn} $\CSPACE(s(n))$ is the set of languages that can be solved by a deterministic catalytic Turing machine that uses at most $s(n)$ size work space and $2^{s(n)}$ size auxiliary space on all inputs $x \in \{0, 1\}^n$. $\CL$ denotes the class $\CSPACE(O(\log n))$.
\end{defn}

For two binary strings of equal length, their Hamming distance is the number of positions at which the corresponding bits are different. We denote the Hamming distance of two binary strings of equal length, say $x$ and $y$, by $\hamdist(x, y)$.

We now formally define the $k$-lossy catalytic Turing machine.

\begin{defn}
For any constant $k$, a \emph{$k$-lossy catalytic Turing machine} is a catalytic Turing machine that, when starting with $w$ in the auxiliary tape, halts with $w'$ in the auxiliary tape such that $\hamdist(w, w') \leq k$.
\end{defn}

\begin{defn}
$\LCSPACE(s(n), k)$ is the set of languages that can be solved by a $k$-lossy catalytic Turing machine using $s(n)$ size work space and $2^{s(n)}$ size auxiliary space, where $k$ is a constant. $k$-$\LCL$ denotes the set $\LCSPACE(O(\log n),k)$.
\end{defn}

We also use the hashing scheme of Fredman, Koml\'{o}s and Szemer\'{e}di \cite{FKS84} in a form useful for us.

\begin{nlemma}
\label{fks}
\cite{FKS84}
Let $S = \{x_1, x_2, \dots, x_k\}$ be a set of $n$-bit integers. Then there exists an ${O}(\log n + \log k)$-bit prime number $p$ so that for all $x_i \neq x_j \in S$, $x_i \bmod{p} \neq x_j \bmod{p}$.
\end{nlemma}

\section{$k$-lossy Catalytic Computation}
\label{lossy:main}
In this section, we will prove our main result, that is, $k$-$\LCL = \CL$.

\begin{thm}
$k$-$\LCL = \CL$.
\end{thm}
\begin{proof}
First note that proving $k$-$\LCL \subseteq \CL$ is sufficient because $\CL \subseteq k$-$\LCL$ follows trivially from the definitions. Let $\mathcal{M}$ be a $k$-lossy catalytic Turing machine with $c\log n$ size work space and $n^c$ size auxiliary space. We will prove $k$-$\LCL \subseteq \CL$ by showing that there exists a deterministic catalytic Turing machine $\mathcal{M'}$ with $c'\log n$ size work space and $n^{c'}$ size auxiliary space, where $c'$ is a sufficiently larger constant than $c$, such that on every input $x$ and initial auxiliary content $w'$, $\mathcal{M}$ accepts $x$ if and only if $\mathcal{M'}$ accepts $x$.

We present the algorithm of $\mathcal{M'}$ in Algorithm \ref{loss:mainalgo}.

\begin{algorithm}[H]
\caption{Algorithm of $\mathcal{M'}$}
\label{loss:mainalgo}
$P$ is the set of $O(\log n)$-bit primes. $I_{2k}$ is the set of tuples of $2k$ indices and $I_{k}$ is the set of tuples of $k$ indices, where each index ranges from 0 to $n^c$.
\begin{algorithmic}[1]
\Procedure{\textsc{losslessSimulation}}{\textrm{Input $x$, Auxiliary Content $w'$}}
\State $good\_prime \gets 0$ \label{hashbegin}
\For{$p \in P$} \label{loss:allprimes}
	\State $bad\_prime\_found \gets$ FALSE
	\For{$\langle ind_1, ind_2 \rangle \in I_{2k} \times I_{2k}$} \label{loss:pairs}
		\If {$w_{ind_1} \neq w_{ind_2}$ and $w_{ind_1}$ mod $p = w_{ind_2}$ mod $p$} \label{loss:checkprime}
		\State $bad\_prime\_found \leftarrow $ TRUE
		\State Jump to line \ref{loss:yy}
		\EndIf
	\EndFor
	
	\If {$bad\_prime\_found =$ FALSE} \label{loss:yy}
		\State $good\_prime \gets p$ \label{loss:goodprime}
		\State Jump to line \ref{loss:storehashvalue}
	\EndIf	
\EndFor \label{hashend}
\\

\State $init\_aux\_val \gets w$ mod $good\_prime$ \label{loss:storehashvalue}
\State Simulate $\mathcal{M}$ on $(x,w)$
\If {$\mathcal{M}$ accepts $x$ in the simulation}
    \State $result \gets$ TRUE 
\Else
    \State $result \gets$ FALSE 
\EndIf
\\
\label{loss:restore} \State Let $z$ be the first $n^c$ bits of auxiliary tape of $\mathcal{M'}$ after the simulation
\For{$\langle ind \rangle \in I_{k}$} \label{loss:restore}
	\If {$z_{ind}$ mod $ good\_prime$ = $init\_aux\_val$} \label{loss:checkrestore}
		\State replace $z$ by $z_{ind}$
		\State Jump to line \ref{loss:halting}
		\EndIf
\EndFor \label{loss:restoreend}
\\
\If{$result = $ TRUE} \label{loss:halting}
	\State \textbf{Accept}
\Else
	\State \textbf{Reject}
\EndIf
\EndProcedure
\end{algorithmic}
\end{algorithm}

\subsubsection*{Description of Algorithm \ref{loss:mainalgo}}
$\mathcal{M'}$ on input $x$ and auxiliary content $w'$, from line \ref{hashbegin} to \ref{hashend},
first finds a prime that hashes all the auxiliary contents in $W$ to $O(\log n)$-bit numbers injectively, where $n = |x|$ and $W$ is the set of $n^c$ size binary strings with hamming distance at most $2k$ from the first $n^c$ bits of $w'$. It is easy to see that $|W| = O(n^{2kc})$. We say a prime number is a \emph{good prime} if it injectively hashes the members of $W$ to $O(\log n)$-bit numbers. Let $P$ denote the set of all $d \log n$-bit prime numbers, where $d$ is a sufficiently large constant, such that at least one of the primes in $P$ is a good prime. The existence of such a $P$ follows from FKS hashing stated in Lemma \ref{fks}. 

Let $I_{2k} = \{(i_1, i_2, \dots, i_{2k})\}$ denote the set of tuples of $2k$ indices, where an index, say $i_j$, ranges from 0 to $n^c$. $\mathcal{M'}$ goes over all the primes in $P$, and for every prime $p$, it checks whether it is a good prime by comparing hashed values of all possible pairs of auxiliary contents of $W$ in line  \ref{loss:checkprime}. To compare the hashed values of all such pairs, $\mathcal{M'}$ iterates over pairs of tuples of indices from $I_{2k}$. For an element $ind \in I_{2k}$, $\mathcal{M'}$ generates an element of $W$ by simply flipping the $i$th bit of $w$ (initial $n^c$ bits of $w'$), if $i \in ind$, for $i \in [1,n^c]$. In line \ref{loss:checkprime}, we denote $w$ with flipped bits according to the tuples of indices, say $ind_1$ and $ind_2$, by $w_{ind_1}$ and $w_{ind_2}$. If all auxiliary contents in $W$ are mapped distinctly, then the loop of the line \ref{loss:pairs} terminates successfully without setting the variable $bad\_prime\_found$ to TRUE. $\mathcal{M'}$ then stores the current prime $p$ in the variable $good\_prime$ in line \ref{loss:goodprime} and terminates the loop of line \ref{loss:allprimes}. Since all the primes and tuples are of $O(\log n)$ size, these iterations can be done in $O(\log n)$ space.

In line \ref{loss:storehashvalue}, $\mathcal{M'}$ stores the hashed value of $w$ in $init\_aux\_val$. After which, it simulates $\mathcal{M}$ on $(x,w)$ and stores the result of acceptance in variable $result$. From line 25 to \ref{loss:restoreend}, $\mathcal{M'}$ restores the initial auxiliary content completely. Let $z$ be the initial $n^c$ bits of the auxiliary tape of $\mathcal{M'}$ after completing the simulation of $\mathcal{M}$ on $(x,w)$. It is possible that $z$ is not equal to $w$ as $\mathcal{M}$ is a lossy catalytic machine. Let $Z$ be the set of all $n^c$ size binary strings whose hamming distance is at most $k$ from $z$. Clearly, $w \in Z$ and $Z \subseteq W$. Also, $good\_prime$ injectively maps the elements of $Z$ to $O(\log n)$ bit primes. To restore $w$, $\mathcal{M'}$ goes over all the binary strings in $Z$ with the help of $I_k$ in the loop of line 26. When it finds the binary string whose hashed value matches $init\_aux\_cal$, it replaces $z$ with that binary string and jumps to line \ref{loss:halting}.

Finally, in line \ref{loss:halting}, $\mathcal{M'}$ checks the value of $result$ and accepts or rejects accordingly.
\end{proof}

\singlespacing

\bibliography{main}

\end{document}